\documentclass[runningheads,a4paper]{llncs}

\usepackage{fullpage}
\usepackage{amsmath,amsfonts,amssymb,amscd}
\usepackage{times}
\usepackage{hyperref}

\usepackage{graphicx}
\usepackage{subfigure}
\usepackage{algorithmic}
\usepackage{algorithm}
\usepackage{cite}

\usepackage[capitalise]{cleveref}
\Crefname{theorem}{Theorem}{Theorems}

\usepackage{soul}
\usepackage{url}
\usepackage{microtype}
\usepackage{bm}
\usepackage{mathrsfs}
\usepackage{txfonts}
\usepackage{graphics}
\usepackage{float}
\usepackage{epsfig}
\usepackage{commath}
\usepackage{subfigure}
\usepackage{caption}

\newcommand{\inlinenorm}[1]{\|#1\|}

\newcommand{\TV}[1]{\norm{#1}_{TV}}
\newcommand{\spann}[1]{\textnormal{span}\left\{#1\right\}}

\newcommand{\R}{\mathbb{R}}
\newcommand{\D}[1]{{\mathcal{D}}_{#1}}
\newcommand{\A}{{A}}
\newcommand{\M}{{M}}
\newcommand{\x}{{x}}
\newcommand{\y}{{y}}
\newcommand{\U}{{U}}
\newcommand{\V}{{V}}
\newcommand{\FNorm}[1]{\norm{#1}_F}
\newcommand{\FKV}{\textsc{FKV }}

\usepackage{pgfplots}
\usepgfplotslibrary{polar}
\usepgflibrary{shapes.geometric}
\usetikzlibrary{trees}
\usetikzlibrary{calc}

\usepackage{booktabs}
\usepackage{commath}
\usepackage{tikz}
\usepackage{tikz-3dplot}

\DeclareFontFamily{U}{lasy}{}
\DeclareFontShape{U}{lasy}{m}{n}{
  <-11> lasy5
  <11-13> lasy6
  <13-15> lasy7
  <15-17> lasy8
  <17-19> lasy9
  <19-> lasy10
}{}

\usetikzlibrary{calc,decorations.pathmorphing,shapes}

\newcounter{sarrow}
\newcommand\xrsquigarrow[1]{%
\stepcounter{sarrow}%
\mathrel{\begin{tikzpicture}[baseline= {( $ (current bounding box.south) + (0,+2ex) $ )}]
\node[inner sep=.5ex] (\thesarrow) {$\scriptstyle #1$};
\path[draw,<-,decorate,
  decoration={zigzag,amplitude=0.7pt,segment length=1.2mm,pre=lineto,pre length=4pt}]
    (\thesarrow.north west) -- (\thesarrow.north east);
\end{tikzpicture}}%
}

\title{A Quantum-inspired Classical Algorithm for Separable Non-negative Matrix Factorization}

\titlerunning{}

\author{Zhihuai Chen$^{1,2}$  \and Yinan Li$^{3}$ \and Xiaoming Sun$^{1,2}$ \and Pei Yuan$^{1,2}$ \and Jialin Zhang$^{1,2}$}
\authorrunning{Zhihuai Chen  \and Yinan Li \and Xiaoming Sun \and Pei Yuan \and Jialin Zhang}

\usepackage{url}

\institute{$^1$CAS Key Lab of Network Data Science and Technology, Institute of Computing Technology, Chinese Academy of Sciences, 100190, Beijing, China \\
$^2$University of Chinese Academy of Sciences, 100049, Beijing, China\\
$^3$Centrum Wiskunde \& Informatica and QuSoft, Science Park 123, 1098XG Amsterdam, Netherlands\\
\path|{chenzhihuai,sunxiaoming,yuanpei}@ict.ac.cn, yinan.li@cwi.nl|
}

\begin{document}

\maketitle
\thispagestyle{empty}
\begin{abstract}\normalsize
    Non-negative Matrix Factorization (NMF) asks to decompose a (entry-wise) non-negative matrix into the product of two smaller-sized nonnegative matrices, which has been shown intractable in general.
    In order to overcome this issue, separability assumption is introduced which assumes all data points are in a conical hull. This assumption makes NMF tractable and is widely used in text analysis and image processing, but still impractical for huge-scale datasets.
    In this paper, inspired by recent development on dequantizing techniques, we propose a new classical algorithm for separable NMF problem. Our new algorithm runs in polynomial time in the rank and logarithmic in the size of input matrices, which achieves an exponential speedup in the low-rank setting.
\end{abstract}


\section{Introduction}
\label{sec:intro}
Non-negative Matrix Factorization (NMF) aims to approximate a non-negative data matrix $\A\in \R_{\ge 0}^{m\times n}$ by the product of two non-negative low rank factors, \emph{i.e.}, $\A \approx WH^T$,
where $W\in\mathbb{R}_{\ge 0}^{m\times k}$ is called basis matrix, $H\in\mathbb{R}_{\ge 0}^{n\times k}$ is called encoding matrix and $k\ll \min\{m, n\}$.
In many applications, an NMF often results in more natural and interpretable part-based decomposition of data \cite{lee1999learning}. Therefore, NMF has been widely used in a number of practical applications, such as topic modeling in text, signal separation, social network, collaborative filtering, dimension reduction, sparse coding, feature selection and hyperspectral image analysis. 
Since computing an NMF is NP-hard \cite{vavasis2009complexity}, a series of heuristic algorithms have been proposed \cite{lee2001algorithms,lin2007projected,hsieh2011fast,kim2008toward,ding2010convex,guan2012mahnmf}. All of the heuristic algorithms aim to minimize the reconstruction error, the formula which is a non-convex program and lack optimality guarantee:

\[
\min_{W\in\R_{\ge 0}^{m\times k}H\in\R_{\ge0}^{n \times k}}\FNorm{\A-WH^T}.
\]

A natural assumption on the data called \emph{separability assumption}, was observed in \cite{donoho2004does} . From a geometry perspective, the separable assumption means that all rows of $\A$ reside in a cone generated by a rather smaller number of rows. In particular, these generators are called anchors of $\A$.
To solve the Separable Non-Negative Matrix Factorizations (SNMF), it is sufficient to identify the anchors in the input matrices, which can be solved in polynomial time \cite{arora2012computing,arora2012learning,gillis2014fast,esser2012convex,elhamifar2012see,zhou2013divide,zhou2014divide}.
Separability assumption is favored by various practical applications. For example, in the unmixing task in hyperspectral imaging, separability implies the existence of `pure' pixel \cite{gillis2014fast}. And in the topic detection task, it also means some words are associated with unique topic \cite{hofmann2017probabilistic}.
In huge datasets, it is useful to pick up some representative data points to stand for other points. Such `self-expression' assumption helps to improve the data analysis procedure
\cite{mahoney2009cur,elhamifar2009sparse}.

\subsection{Related work}
 It is natural to assume all the rows of the input $\A$ has unit $\ell_1$-norm, since $\ell_1$-normalization translates the conical hull to convex hull while keeping the anchors unchanged.
 From this perspective, most algorithms essentially identify the extreme points in the convex hull of the ($\ell_1$-normalized) data vectors.
 In \cite{arora2012computing}, the authors use $m$ linear programs in $O(m)$ variables to identify the anchors out of $m$ data points, and it is therefore not suitable for dealing with large-scale real-world problems.
 Furthermore, \cite{recht2012factoring} presents a single LP in $n^2$ variables for SNMF to deal with large-scale problems (but is still impractical for huge-scale problems). 


There is another class of algorithms based on greedy algorithms. The main idea is to opt a data point on the direction where the current residual decreases fast. The algorithms terminate with a sufficiently small error or a large iteration times. For example, Successful Projection Algorithm (SPA) \cite{gillis2014fast} derives from Gram-Schmidt orthogonalization with row or column pivoting. XRAY \cite{kumar2013fast} detects a new anchor referring to the residual of exterior data points and updates the residual matrix by solving a nonnegative least square regression. Both of these two algorithms based on greedy pursuit have smaller time complexity compared with LP-based methods. However, the time complexity is still too large for large-scaled data.

\cite{zhou2013divide,zhou2014divide} utilize a Divide-and-Conquer Anchoring (DCA) framework to tackle the SNMF. Namely, by projecting the data set into several low-dimension subspaces, and each projection can determines a small set of anchors. Moreover, it can be proven that all the $k$ anchors can be identified by $O(k\log k)$ projections.

Recently, a quantum algorithm for SNMF called Quantum Divide-and-Conquer Anchoring algorithm (QDCA), has been presented \cite{du2018quantum}, which uses the quantum technology to speed up the random projection step in \cite{zhou2013divide}.
QDCA implements matrix-vector product (\textit{i.e.}, random projection) via quantum principal component analysis and then a quantum state encoding the projected data points could be prepared efficiently.
Moreover, there are also several papers utilizing dequantizing techniques to solve some low-rank matrix operations, such as recommendation systems \cite{tang2018recommendation} and matrix inversion \cite{gilyen2018quantum,chia2018quantum}. Dequantizing techniques in those algorithms involve two technologies, the Monte-Carlo singular value decomposition and rejection sampling, which could efficiently simulate some special operations on low-rank matrices.

Inspired by QDCA and the dequantizing techniques , we propose a classical randomized algorithm which speeds up the random projection step in \cite{zhou2013divide} and thereby identifies all anchors efficiently. Our algorithm takes time polynomial in rank $k$, condition number $\kappa$ and {\em logarithm} of the size of matrix. When rank $k=O(\log (mn))$, our algorithm achieves exponentially speedup than any other classical algorithms for SNMF.

\subsection{Organizations}
The rest of this paper is organized as follows: In Section \ref{sec:pre}, we introduce notations, models and preliminaries of our algorithm; in Section \ref{sec:fastdca}, we present our algorithm and analyze its correctness and running time; and Section \ref{sec:con} concludes with a discussion of this paper and the future work.

%

\section{Preliminaries}
\label{sec:pre}
\subsection{Notations}
 Let $[n]:=\{1,2,\ldots,n\}$. Let
$\text{span}\{\x_i\in\R^n | i\in[k]\}:=\{\sum_{i=1}^k\alpha_i\x_i | \alpha_i\in\R, i\in[k]\}$ denote the space spanned by $\x_i$ for $i\in[k]$.
 For a matrix $\A\in\R^{m\times n}$, $\A_{(i)}$ and $\A^{(j)}$ denote the $i$th row and the $j$th column of $\A$ for $i\in[m],j\in[n]$, respectively.
 Let $\A_R=[\A^T_{(i_1)},\A^T_{(i_2)},\ldots,\A^T_{(i_r)}]^T$ where $\A\in\R^{m\times n}$ and $R=\{i_1,i_2,\ldots,i_r\}\subseteq [m]$ (without loss of generality, assume $i_1\le i_2\le\cdots\le i_r$).  $\FNorm{\A}$ and $\norm{\A}_2$ refer to Frobenius norm and spectral norm, respectively.
 For a vector $v\in\R^n$, $\norm{v}$ denotes its $\ell_2$-norm.
For two probability distributions $p,q$ (as density functions) over a discrete universe $D$, the \emph{total variation distance} between them is defined as $\norm{p,q}_{TV}:=\frac{1}{2}\sum_{i\in D}|p(i)-q(i)|$.
$\kappa(\A):=\sigma_{\max}/\sigma_{\min}$ denotes the condition number of $\A$, where $\sigma_{\max}$ and $\sigma_{\min}$ are the maximal and minimal {\em non-zero} singular values of $\A$.

\subsection{Sample Model}
\label{sec:model}


In query model, algorithms for SNMF problem require time which is at least linear in the number of nonzero elements of the matrix, since in the worst case, they have to read out all entries. However, we expect our algorithm to be efficient even if the datasets are extremely large.  Considering the QDCA in \cite{du2018quantum}, one of its advantage is that data is prepared in quantum state and can be access via `quantum' way (like sampling). Thus, in quantum algorithm, quantum state is served to represent data implicitly which can be read out by measurement only. In order to avoiding reading the whole matrix, we introduce a new sample model other than the query model based on the idea of quantum state preparation assumption.
\begin{definition}[$\ell_2$-norm Sampling]
Let $\D{v}$ denote the distribution over $[n]$ with density function $\D{v}(i)=v_i^2/\norm{v}^2$ for $v\in\R^n$. A sample from a distribution $\D{v}$ is called a sample from $v$.
\end{definition}
\begin{lemma}[Vector Sample Model]\label{lem:vec_sample}
    There is a data structure storing vector $v\in \R^n$ in $O(n\log n)$ space, and supporting following operations:
    \begin{itemize}
        \item Querying and updating a entry in $O(\log n)$ time;
        \item Sampling from $\D{v}$ in $O(\log n)$ time;
        \item Finding $\norm{v}$ in $O(1)$ time.
    \end{itemize}
\end{lemma}

Such a data structure can be easily implemented via Binary Search Tree (BST) (see Figure~\ref{fig:bst}).
\begin{figure}[H]
    \centering
    \begin{tikzpicture} [level distance=0.8cm,
        level 1/.style={sibling distance=3cm},
        level 2/.style={sibling distance=1.5cm}]
        \node {$\norm{v}^2$}
            child {node {$v_1^2+v_2^2$}
                child {node {$v_1^2$} child{node {$\textnormal{sgn}(v_1)$}}}
                child {node {$v_2^2$} child{node {$\textnormal{sgn}(v_2)$}}}
            }
            child {node {$v_3^2+v_4^2$}
                child {node {$v_3^2$} child{node {$\textnormal{sgn}(v_3)$}}}
                child {node {$v_4^2$} child{node {$\textnormal{sgn}(v_4)$}}}
            };
    \end{tikzpicture}
    \caption{Binary search tree for $v=(v_1,v_2,v_3,v_4)^T\in \R^4$. The leaf node stores $v_i^2$ and interior node stores the sum of the children. In order to restore the original vector, we also store the sign of $v_i$ in leaf node. To sampling from $\D{v}$, we can start from top and randomly recurring on a child, with probability proportional to its weight.}\label{fig:bst}
\end{figure}
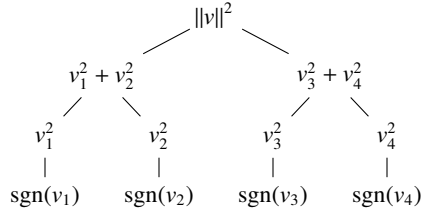

\begin{proposition}[Matrix Sample Model]\label{prop:matrix_model}
    Considering matrix $\A\in \R^{m\times n}$, let $\tilde{\A}$ and $\tilde{\A}'$ be the vector whose entry is $\norm{\A_{(i)}}$ and $\norm{\A^{(j)}}$, respectively.
    There is a data structure storing matrix $\A\in \R^{m\times n}$ in $O(mn)$ space and supporting following operations:
    \begin{itemize}
        \item Querying and updating an entry in $O(\log m + \log n)$ time;
        \item Sampling from $\A_{(i)}$ for any $i\in [m]$ in time $O(\log n)$;
        \item Sampling from $\A^{(j)}$ for any $j\in [n]$ in time $O(\log m)$;
        \item Finding $\FNorm{\A}$, $\norm{\A_{(i)}}$ and $\norm{\A^{(j)}}$ in time $O(1)$;
        \item Sampling $\tilde{\A}$ and $\tilde{\A}'$ in time $O(\log m)$ and $O(\log n)$, respectively.
    \end{itemize}
\end{proposition}

This data structure can be easily implemented via Lemma~\ref{lem:vec_sample}, we can just use two arrays of BST to store all rows and columns of $\A$ and use two extra BSTs store $\tilde{\A}$ and $\tilde{\A}'$.

\subsection{Low-rank Approximations in Sample Model}


\FKV algorithm is a Monte-Carlo algorithm \cite{frieze2004lowrank} that returns approximate singular vectors of given matrix $\A$ in {\em matrix sample model}. The low-rank approximation of $\A$ can be reconstructed by approximate singular vectors.
The query and sample complexity of FKV algorithm are independent of size of $\A$. FKV algorithm outputs a short `description' of $\hat{\V}$, which is approximate to a right singular vectors $\V$ of matrix $\A$.
Similarly, \FKV algorithm can output a description of approximate left singular vectors $\hat{\U}$ of $\A$ by inputting $\A^T$.
Let \FKV($\A,k,\epsilon,\delta$) denote the \FKV algorithm, where $\A$ is a matrix given by sample model, $k$ is the rank of approximate matrix of $\A$, $\epsilon$ is error parameter, and $\delta$ is the failure probability. The \FKV algorithm is described in Theorem \ref{thm:FKV}.

\begin{theorem}[Low-rank Approximations, \cite{frieze2004lowrank}]\label{thm:FKV}
    Given matrix $\A\in\R^{m\times n}$ in matrix sample model, $k\in \mathbb{N}$ and $\epsilon, \delta\in(0,1)$, \FKV algorithm outputs the description of the approximate right singular vectors $\hat{\V}\in\R^{n\times k}$ in $O(poly(k,1/\epsilon,\log\frac{1}{\delta}))$ samples and queries of $\A$ with probability $1-\delta$, which satisfies
    \[
        \FNorm{\A\hat{\V}\hat{\V}^T - \A}^2 \leq \min_{D:rank(D)\leq k}\FNorm{\A-D}^2 + \epsilon\FNorm{\A}^2.
    \]
\end{theorem}

Especially, if $\A$ is a matrix with rank $k$ exactly, Theorem~\ref{thm:FKV} also implies an inequality:
$$\FNorm{\A\hat{\V}\hat{\V}^T - \A} \leq \sqrt{\epsilon}\FNorm{\A}.$$

\paragraph{Description of $\hat{\V}$.} Note that \FKV algorithm does not output the approximate right singular vectors $\hat{\V}$ directly since their lengths are linear of $n$.
It returns a description of $\hat{\V}$, which consists of three components:
the row index sets $T:=\{i_t\in[m]|t\in[p]\}$, a vector set $U:=\{u^{(j)}\in\R^{p}|j\in [k]\}$ which are singular vectors of a submatrix sampled from $\A$ , and its corresponding singular values $\Sigma:=\{\sigma^{(j)}|j\in[k]\}$, where $p=O(poly(k,\frac{1}{\epsilon}))$. In fact, $\hat{\V}^{(i)}:=\A_T u^{(i)}/\sigma_i$ for $i\in [k]$. Given a description of $\hat{\V}$, we can {\em sample} from $\hat{\V}^{(i)}$ in time $O(poly(k,\frac{1}{\epsilon}))$ for $i\in[k]$
\cite{tang2018recommendation} and {\em query} its entry in time $O(poly(k,\frac{1}{\epsilon}))$.

\begin{definition}[$\alpha$-orthonormal]
    Given $\alpha>0$, $\hat{\V}\in\R^{n\times k}$ is called $\alpha$-approximately orthonormal if $1-\alpha/k\le\norm{\hat{\V}^{(i)}}^2\le1+\alpha/k$ for $i\in[k]$ and $|\hat{\V}^{(s)}\hat{\V}^{(t)}|\le \alpha/k$ for $s\neq t\in[k]$.
\end{definition}
The next lemma presents some properties of $\alpha$-approximate orthonormal vectors.
\begin{lemma}[Properties of $\alpha$-orthonormal Vectors, \cite{tang2018recommendation}]\label{lem:orthonormal}
    Given a set of  $k$ $\alpha$-approximately orthonormal vectors $\hat{\V}\in\R^{n\times k}$, then there exists a set of $k$ orthonormal vectors $\V\in\R^{n\times k}$ spanning the columns of $\hat{\V}$ such that
    \begin{align}
    & \FNorm{\V-\hat{\V}} \leq \alpha/\sqrt{2} + c_1\alpha^2,\label{eq:orthonormal}\\
    & \FNorm{\Pi_{\hat{\V}} - \hat{\V}\hat{\V}^T} \leq c_2\alpha,\label{eq:app_proj}
    \end{align}
    where $\Pi_{\hat{\V}}:=\V\V^T$ represents the orthonormal projector to image of $\hat{\V}$ and $c_1,c_2>0$ are constants.
\end{lemma}
\begin{lemma}[\cite{frieze2004lowrank}]\label{FKVorth}
The output vectors $\hat{\V}\in\R^{n\times k}$ of $\FKV(\A,k,\epsilon,\delta)$ is $\epsilon k/16$-approximate orthonormal.
\end{lemma}

\section{Fast Anchors Seeking Algorithm }
\label{sec:fastdca}
In this section, we present a randomized algorithm for SNMF which is called Fast Anchors Seeking (FAS) Algorithm. Especially, the input $\A\in\R_{\ge 0}^{m\times n}$ of FAS is given by matrix sample model which is realized via a data structure described in Section \ref{sec:pre}.
FAS returns the indices of anchors in time polynomial logarithmic to the size of matrix.

\subsection{Description of Algorithm}

Recall that SNMF aims to factorize $\A= F\A_R$ where $R$ is the index set of anchors. In this paper, an additional constraint is added: the sum of entries in any row of $F$ is 1. Namely, any data point of $\A$ resides in convex hull which is the set of all convex combination of $\A_R$. In fact, normalizing each row of matrix $\A$ by $\ell_1$-norm is valid, since the anchors remain unchanged. Moreover, Instead of storing $\ell_1$-normalized matrix $A$, we can just maintain the $\ell_1$-norms for all rows and columns.

The Quantum Divide-and-Conquer Anchoring (QDCA) is a quantum algorithm for SNMF which achieves exponential speedup than any classical algorithms \cite{du2018quantum}. After projecting any convex hull into an 1-dimensional space, the geometric information is partially preserved. Especially, the anchors in 1-dimensional projected subspace are still anchors in the original space. The main idea of QDCA is quantizing random projection step in DCA.
It decomposes SNMF into several subproblems: projecting $\A$ onto a set of random unit vectors $\{\beta_i\in\mathbb{R}^n\}_{i=1}^{s}$ with $s=O(k\log k)$, \emph{i.e.}, computing $\A\beta_i\in\R^{m}$. Such a matrix-vector product can be efficiently implemented by Quantum Principle Component Analysis (QPCA). And then it returns a $\log m$-qubits quantum state whose amplitudes are proportional to entries of
$\A\beta_i$. Measurement of quantum state outcomes an index $j\in[m]$ which obeys distribution $\D{\A\beta_i}$.
Thus, we can prepare $O(poly\log m)$ copies of quantum states, measure each of them in computational basis and record the most frequent index. By repeating procedure above with $s=O(k\log k)$ times, we could successively identify all anchors with high probability.

As discussed above, the core and most costly procedure is to simulate $\D{\A\beta_i}$. At the first sight, traditional algorithms can not achieve exponential speedup on account of limits of computational model. In QDCA, vectors are encoded into quantum states and we can sample the entries with probability proportional to their magnitudes by measurements. This quantum state preparation overcomes the bottleneck of traditional computational model. Based on divided-and-conquer scheme and sample model (See Section \ref{sec:model}), we present Fast Anchors Seeking (FAS) Algorithm inspired by QDCA. Designing FAS is quite hard and non-trivial although FAS and QDCA have the same scheme.
Indeed, we can simulate $\D{\A\beta_i}$ directly by rejection sampling technology. However, the number of iterations of rejection sampling is unbounded. To overcome this difficulty, we translate matrix $A$ into its approximation $\hat{U}\hat{U}^TA$, where the columns $\hat{\U}\in \mathbb{R}^{m\times k}$ consists of $k$ approximate left singular vectors of matrix $\A$ and $k=rank(\A)$. Next, it is obvious that $\y=\hat{\U}^TA\beta_i \in \mathbb{R}^k$ is a short vector and we can estimate its entries one by one (see Lemma \ref{cor:inner_estimation}) efficiently. Now the problem becomes to simulate $\D{\hat{\U}\y}$ and it can be done by Lemma~\ref{lem:tang} .

Given an error parameter $\epsilon/2$, the method described above will result in $\norm{\A\beta_i-\hat{\U}\hat{U}^TA\beta_i}<\epsilon \FNorm{A}\norm{\beta_i}/2$ via Theorem \ref{thm:FKV}, which implies $\TV{\D{\A\beta_i}-\D{\hat{\U}\hat{U}^T\A\beta_i}}\le \epsilon \FNorm{A}\norm{\beta_i}/\norm{\A\beta_i}$. Namely, the method above introduces an unbounded error in form $\epsilon \FNorm{A}\norm{\beta_i}/\norm{A\beta_i}$ if $\beta_i$ is  arbitrary vector in entire space $\mathbb{R}^n$.
Fortunately, this issue can be solved by generating random vectors $\{\beta_i\}_{i=1}^s$ lying in row space of $A$ instead of those lying in entire space $\mathbb{R}^n$.
To generate uniform random unit vectors on the row space of $\A$, we need to find a basis of row space of $\A$.
If $\V\in \R^{n\times k}$ is a set of orthonormal basis of the row space of $\A$ (the space spanned by the right singular vectors), and $\x_i$ is uniform random unit vector on $\mathbb{S}^{k-1}$, then $\beta = \V\x_i$ is a unit random vector in row space of $\A$.
Moreover, \FKV algorithm will figure out approximate singular vectors $\hat{\V}$ for $\V$,
that can help us make an approximate $\hat{\beta_i}=\hat{\V}\x_i$ for $\beta_i$.
Therefore, we will estimate distribution $\D{\hat{U}\hat{\U}^T\A\hat{\V}\x_i}$ instead of $\D{\A\beta_i}$.
Based on Corollary~\ref{cor:inner_estimation}, $\hat{\U}^T\A\hat{\V}$ can be estimated efficiently. According to Lemma \ref{lem:tang}, $\hat{\U}\y$ can be sampled efficiently, thus we can treat $\tilde{\y}$ as estimation of $\hat{\U}^T\A\hat{\V}\x_i$ (see Figure~\ref{fig:relation}).
\begin{figure*}[!htb]
    \centering
    \[
        \D{\A\beta_i}
        \xleftarrow{\textnormal{(1)}~\beta_i=\V\x_i}
        \D{\A\V\x_i}
        \overset{\textnormal{(2)}~\text{Lemma}~\ref{lem:orthonormal},~\text{Lemma}~\ref{lem:span}}{\xrsquigarrow{\text{span}\{\V^{(i)}|i\in[k]\}=\text{span}\{\hat{\V}^{(i)}|i\in[k]\}=\text{span}\{A_{(i)}|i\in[m]\}}}
        \D{\A\hat{\V}\x_i}
        \overset{\textnormal{(3)}~\text{Theorem}~\ref{thm:FKV}}{\xrsquigarrow{ \hat{U}\hat{U}^TA \approx A }}
        \D{\hat{\U}\hat{\U}^T\A\hat{\V}\x_i}
        \overset{\textnormal{(4)}~\text{Lemma}~\ref{cor:inner_estimation}}{\xrsquigarrow{~\tilde{\M}\approx\hat{\U}^T\A\hat{\V}, \tilde{\y_i}=\tilde{\M}\x_i}}
        \D{\hat{\U}\tilde{\y_i}}
        \overset{\textnormal{(5)}~\text{Lemma~\ref{lem:tang}}}{\xrsquigarrow{\textnormal{rejection sampling}}}
        \mathcal{O}_i
    \]

    \caption{An illustration for how to approximate distribution $\D{\A\beta_i}$.
    $\mathcal{O}_i$ represents the final distribution which approximate $\D{\A\beta_i}$.
        $\leftsquigarrow$ represents `approximate' and $\leftarrow$ represents `equal to' in a sense of total variation distance.
        To prove the upper bound for $\|\D{\A\beta_i}, \mathcal{O}_i\|_{TV}$,
        we introduce several medium distributions $\D{A\hat{V}x_i}$, $\D{\hat{U}\hat{U}^TA\hat{V}\x_i}$ and $\D{\hat{U}y}$.
        From left to right, (1) use a set of right singular vectors $V\in \mathbb{R}^{n\times k}$ to generate the random unit vector lying in the row space of $A$;
        (2) however, since $V$ cannot be gained efficiently,  use FKV algorithm to figure out an approximation $\hat{V}$ given by a `short desription'; (Lemma \ref{lem:orthonormal},~Lemma \ref{lem:span})
        (3) translate $A\hat{V}x_i$ into $\hat{U}\hat{U}^TA\hat{V}x_i$ since $\A\approx \hat{\U}\hat{U^T}\A$ (Theorem \ref{thm:FKV}), where $\hat{U}$ is the approximate left singular vectors generated by FKV given by a `short description'.
        (4) return an estimation $\tilde{\M}$ of $M=\hat{U}^TA\hat{V} \in \mathbb{R}^{k\times k}$ by estimating each entry by Lemma~\ref{cor:inner_estimation} and then approximate $Mx_i$ (denoted as $\tilde{y}_i$);
        (5) finally, the rejection sampling works for $\hat{U}\tilde{y}_i$ by Lemma~\ref{lem:tang} since $\hat{U}$ is approximately orthonormal.
    }
    \label{fig:relation}
\end{figure*}

Once we can simulate distribution $\D{\A\V\x_i}$, we can figure out the index of the largest component of vector $\A\V\x_i$ by picking up $O(poly\log m)$ samples (Theorem~\ref{thm:sample_largest}). Moreover, according to \cite{zhou2013divide}, by repeating this procedure with $O(k\log{k})$ times, we can find all anchors of $\A$ with high probability (For single step of random projection, see Figure~\ref{fig:our_proj}).

\begin{figure}[h]
    \centering
    \includegraphics[width=10cm]{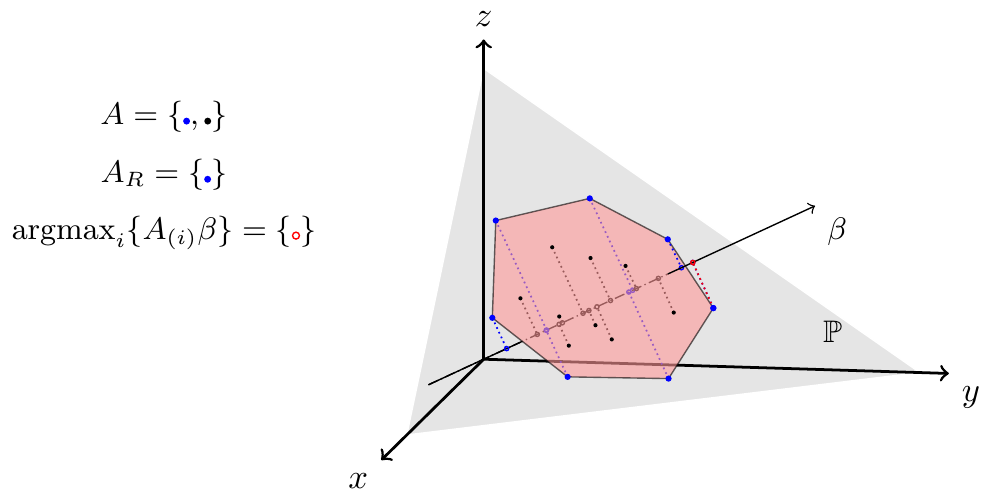}
    \caption{An illustration of finding an anchor of $A$ with rank $k=3$. The 3-dimensional space represents the row space of $A$ and $A$'s data points (blue and black points) lie on the $\ell_1$-normalized plane $\mathbb{P}$.
    The blue points also stand for the anchors of $A$. A random vector $\beta$ is picked up from row space of $A$ and then data points are projected on $\beta$. The anchors of the projected space on $\beta$ are still the anchors of $A$, such that implies that the red point with the maximum absolute projection component on $\beta$ is an anchor of $A$. }
    \label{fig:our_proj}
\end{figure}

\begin{algorithm}[! htb]
    \caption{Fast Anchors Seeking Algorithm}
    \label{alg:fas}
    \textbf{Input}: Separable non-negative matrix $\A\in \R_{\ge 0}^{m\times n}$ in matrix sample model, $k=rank(\A)$, condition number $\kappa$, a constant $\delta\in (0,1)$ and $s=O(k\log k)$.\\
    \textbf{Output}: The index set $R$ of anchors for $\A$.\\
    \begin{algorithmic}[1] 
        \STATE Initialize $R=\emptyset$.
        \STATE Set $\epsilon < 2\sqrt{2\log(4\log^2(m)/\delta)/\log^2(m)}$.
        \STATE Set $\epsilon_{\V}=O\left(\min\left\{\epsilon/\sqrt{k}\kappa,1/k\kappa^2\right\}\right)$, $\delta_\V=1-(1-\delta)^{\frac{1}{4}}$.
        \STATE Run $\FKV(\A, k, \epsilon_\V, \delta_\V)$ and output the description of approximate right singular vectors $\hat{\V}$. \label{ln:fkvv}
        \STATE Set  $\epsilon_{\U}=O\left(\min\left\{\frac{\epsilon}{k},\frac{1}{k\kappa^2}\right\}\right)$, $\delta_\U=1-(1-\delta)^{\frac{1}{4}}$.
        \STATE Run $\FKV(\A^T, k, \epsilon_\U, \delta_\U)$ and output the description of approximate left singular vectors $\hat{\U}$.\label{ln:fkvu}
        \STATE By Lemma \ref{cor:inner_estimation}, estimate $\M:=\hat{\U}^T\A\hat{\V}$ with relative error $\zeta=O\left(\epsilon/k^2\kappa\right)$ and failure probability $\eta=1-(1-\delta)^{\frac{1}{4}}$, and denote the result as $\tilde{\M}$.\label{ln:inner-est}
        \FOR{$i=1$ to $s$}
        \STATE Generate a unit random vector $\x_i\in\R^k$.
        \STATE Directly compute $\tilde{y}_i=\tilde{\M}\x_i$.
        \STATE By rejection sampling (Algorithm \ref{alg:rej}), simulate distribution $\D{\hat{\U}\tilde{y}}$ with failure probability $\gamma=1-(1-\delta)^{\frac{1}{4s}}$ and pick up $O(poly\log m)$ samples. \label{ln:rej}\\
        \COMMENT{Let $\mathcal{O}_i$ denotes the actual distribution which simulates $\D{\hat{\U}\tilde{y}_i}$ }
        \STATE $R\leftarrow R\cup \{ l \}$, where $l$ is the most frequently index appearing in $O(poly\log m)$ samples.
        \ENDFOR
        \STATE Return $R$
    \end{algorithmic}
    \end{algorithm}

\subsection{Analysis}
Now, we propose our main theorem and analyze the correctness and complexity of our algorithm FAS.
\begin{theorem}[Main Result]
    Given separable non-negative matrix $\A\in\R_{\ge 0}^{m\times n}$ in matrix sample model, the rank $k$, condition number $\kappa$ and a constant $\delta\in (0,1)$, Algorithm \ref{alg:fas} returns the indices of anchors with probability at least $1-\delta$ in time
    \[
        O\left(poly\left(k,\kappa,\log \frac{1}{\delta}, \log (mn)\right)\right).
    \]
\end{theorem}

\subsubsection{Correctness}
In this subsection, we will analyze the correctness of Algorithm~\ref{alg:fas}.
Firstly, we show that the columns of $\V$ defined in Lemma~\ref{lem:orthonormal} form a basis of row space of matrix $\A$, which is necessary to generate unit vector in row space of $\A$.
The next, we prove that for each $i\in [s]$, distribution $\mathcal{O}_i$ is $\epsilon$-close to distribution $\D{\A\V\x_i}$ in total variant distance.
Once again, we show how to gain the index of largest component of $\A\V\x_i$ from distribution $\mathcal{O}_i$. Finally, by $O(k\log{k})$ random projection, it is enough for us to gain all anchors of matrix $\A$.

The following lemma tells us the approximate singular vectors outputted by \FKV spans the row space of matrix $\A$. And combining with Lemma~\ref{lem:orthonormal}, it gives us that $\V$ also spans the same space, i.e., $\V$ forms an orthonormal basis of row space of matrix $\A$.
\begin{lemma}\label{lem:span}
    Let $\hat{\V}$ be the output of algorithm $\FKV(\A, k, \epsilon, \delta)$. If $\epsilon < \frac{1}{k\kappa^2}$, then with probability $1-\delta$, we obtain
    \[
        \spann{\hat{\V}^{(i)}\middle|i\in[l],l\leq k} = \spann{\A_{(i)}\middle|i\in[m]}.
    \]
\end{lemma}

\begin{proof}
    By contradiction, we assume that $\spann{\hat{\V}^{(i)}\middle|i\in[k]} \neq \spann{\A_{(i)}\middle|i\in[m]}$, which implies that there exists a unit vector $\x \in \spann{\A_{(i)}\middle|i\in[m]}$ and $\x \perp \spann{\hat{\V}^{(i)}\middle|i\in[k]}$. Then we can obtain $\norm{\A\x-\A\hat{\V}\hat{\V}^T\x}=\norm{\A\x}\ge \sigma_{\min}(\A)$ since $\hat{\V}^T\x = \vec{0}$.
    And according to Theorem~\ref{thm:FKV}, we have
    \[
        \norm{\A\x-\A\hat{\V}\hat{\V}^T\x}\le \sqrt{\epsilon}\FNorm{\A}\le \sqrt{\epsilon k}\kappa\sigma_{\min}(\A).
    \]
    Thus $\sigma_{\min}(\A)\le \sqrt{\epsilon k}\kappa\sigma_{\min}(\A)$,  which makes a contradiction if $\epsilon<1/k\kappa^2$.
\end{proof}
By Lemma \ref{lem:span}, we can generate an approximate random vector in the row space of $\A$ with probability $1-\delta$ in time $O(poly(k,1/\epsilon,\log 1/\delta))$ by \FKV($\A,k,\epsilon,\delta$). Firstly, we obtain the description of approximate right singular vectors by $\FKV$ algorithm, where the error parameter $\epsilon$ is bounded by rank $k$ and condition number $\kappa$ (see in Lemma \ref{lem:span}). Secondly, we generate a random unit vector $\x_i\in\R^k$ as a coordinate vector
referring to a set of orthonormal vectors in Lemma \ref{lem:orthonormal}. Let $V$ denotes the matrix defined in Lemma~\ref{lem:orthonormal}, then it is obvious that its columns form the right singular vectors for matrix $A$. That is, $\hat{\beta_i}=\hat{\V}\x_i$ is an approximate vector of a random vector $\beta=\V\x_i$. 
Next, we show that total variant distance between $\mathcal{O}_i$ and $\D{\A\V\x_i}$ is bounded by constant $\epsilon$.
For convenience, we assume that each step in Algorithm~\ref{alg:fas} succeeds and the final success probability will be given in next subsection.

\begin{lemma}\label{lem:distr_appro}
    For all $i \in [s]$, $\TV{\mathcal{O}_i, \D{\A\V\x_i}}\leq \epsilon$ holds simultaneously with probability $1-\delta$.
\end{lemma}

In the rest, without ambiguity, we use notations $\mathcal{O}$, $\x$ instead of $\mathcal{O}_i$, $\x_i$. By applying triangle inequality, we divide the left part of inequality into four parts (the intuition idea please ref Figure~\ref{fig:relation}):
\begin{align*}
    \TV{\D{\A\V\x}, \mathcal{O}} & \leq \underbrace{\TV{\D{\A\V\x}, \D{\A\hat{\V}\x}}}_{\textcircled{1}}
    + \underbrace{\TV{\D{\A\hat{\V}\x}, \D{\hat{\U}\hat{\U}^T\A\hat{\V}\x}}}_{\textcircled{2}}
                                 + \underbrace{\TV{\D{\hat{\U}\hat{\U}^T\A\hat{\V}\x}, \D{\hat{U}\tilde{\M}\x}}}_{\textcircled{3}}
                                 + \underbrace{\TV{\D{\hat{\U}\tilde{\M}\x}, \mathcal{O}} }_{\textcircled{4}}.
\end{align*}

Thus, we only need to prove that \textcircled{1}, \textcircled{2}, \textcircled{3}, \textcircled{4} $<\frac{\epsilon}{4}$, respectively.
In addition, given $u,v\in\R^n$, if $\norm{u-v}\le \frac{\epsilon}{2}\norm{u}$, then $\TV{\D{u},\D{v}}\le \epsilon$.
For \textcircled{1}, \textcircled{2} and \textcircled{3}, we only show their $\ell_2$-norm version, \emph{i.e.},
\begin{itemize}
    \item $\inlinenorm{\A\V\x - \A\hat{\V}\x}\le \frac{\epsilon}{8} \inlinenorm{\A\V\x}$;
    \item $\inlinenorm{\A\hat{\V}\x - \hat{\U}\hat{\U}^T\A\hat{\V}\x} \le \frac{\epsilon}{8}\inlinenorm{\A\hat{\V}\x}$;
    \item $\inlinenorm{\hat{\U}\hat{\U}^T\A\hat{\V}\x - \hat{\U}\tilde{\M}\x} \le \frac{\epsilon}{8}\inlinenorm{\hat{\U}\hat{\U}^T\A\hat{\V}\x}$.
\end{itemize}

For convenience, in the rest part, let $\alpha_\U = \epsilon_\U k / 16$ and $\alpha_\V = \epsilon_\V k / 16$ represent approximate ratio for orthonormality of $\hat{\U}$ and $\hat{\V}$ based on Lemma \ref{FKVorth}, respectively.

Before we start our proof, we list two tools which are used to prove \textcircled{3} and \textcircled{4}, respectively.

Based on rejection sampling, Lemma \ref{lem:tang} shows that sampling from linear combination of $\alpha$-approximately orthogonal vectors can be quickly realized without knowledge of norms of these vectors (see Algorithm~\ref{alg:rej}).
\begin{lemma}[\cite{tang2018recommendation}]\label{lem:tang}
    Given a set of $\alpha$-approximately orthonormal vectors $\hat{\V}\in \R^{n\times k}$ in vector sample model, and an input vector $w\in\R^{k}$, there exists an algorithm outputting a sample from a distribution $\frac{\alpha}{1-\alpha}$-close to $\D{\hat{\V}w}$ with probability $1-\gamma$ using $O(k^2\log\frac{1}{\gamma}(1+O(\alpha)))$ queries and samples.
\end{lemma}
\begin{algorithm}[tb]
    \caption{Rejection Sampling for $\D{\hat{U}y}$}
    \label{alg:rej}
    \textbf{Input}: A set of approximately orthonormal vectors $\hat{U}^{(j)}\in \mathbb{R}^{m}$ (for $j=1,\dots, k)$ in {\em vector sample model} and a vector $y\in \mathbb{R}^k$.\\
    \textbf{Output}: A sample $s$ subjecting to $\D{\hat{U}y}$.\\
    \begin{algorithmic}[1] 
        \STATE Sample $j$ according to probability proportional to $\|\hat{U}^{(j)} y_j\|$.
        \STATE Sample $s$ from $\D{\hat{U}^{(j)}}$.
        \STATE Compute $r_s=\frac{(\hat{U}y)_s^2}{k\sum_i(y_i\hat{U}_{ij})}$.
        \STATE Accept $s$ with probability $r_s$, otherwise, restart.
    \end{algorithmic}
\end{algorithm}




\begin{lemma}\label{cor:inner_estimation}
    Given $\A\in \mathbb{R}^{m\times n}$ in {\em matrix sample model} and  $L\in\mathbb{R}^{k_1\times m}$ and $R\in\mathbb{R}^{n\times k_2}$ in query model, let $\M=L\A R$, then we can output a matrix $\tilde{\M}\in \mathbb{R}^{k_1\times k_2}$, with probability $1-\eta$, such that
    \[
        \FNorm{\M-\tilde{\M}}\leq \zeta\FNorm{\A}\FNorm{L}\FNorm{R}
    \]
    by $O\left(k_1k_2\frac{1}{\zeta^2}\log \frac{1}{\eta} \right)$ queries and samples.
\end{lemma}
\vspace{-0.2cm}
\begin{proof}
    Let $\M_{ij}=L_{(i)}\A R^{(j)}$ with $i\in[k_1]$ and $j\in[k_2]$. In \cite{tang2018recommendation}, there exists an algorithm that outputs an estimation of $\M_{ij}$ ($\tilde{\M}_{ij}$) to precision $\zeta\FNorm{\A}\norm{L_{(i)}}\norm{R^{(j)}}$ with probability $1-\eta'$ in time $O\left( \frac{1}{\zeta^2}\log \frac{1}{\eta'} \right)$. Let $\eta'=1-(1-\eta)^{1/(k_1k_2)}$. We can output $\tilde{\M}$ with probability $1-\eta$ utilizing $O\left(k_1k_2\frac{1}{\zeta^2}\log \frac{1}{1-(1-\eta)^{1/k^2}}\right)=O\left(k_1k_2\frac{1}{\zeta^2}\log \frac{1}{\eta}\right)$ queries and samples respectively where $\tilde{\M}$ satisfies
    \begin{align*}
        \FNorm{\M-\tilde{\M}}^2 &=\sum_{ i\in[k_1],j\in[k_2]}|\M_{ij}-\tilde{\M}_{ij}|^2 \\
                                &\le \sum_{i\in[k_1],j\in[k_2]}\zeta^2\FNorm{\A}^2\norm{L_{(i)}}^2\norm{R^{(j)}}^2 \\
                                &= \zeta^2 \FNorm{\A}^2 \sum_{i\in[k_1]}\norm{L_{(i)}}^2 \sum_{j\in[k_2]} \norm{R^{(j)}}^2 \\
                                &=\zeta^2 \FNorm{\A}^2 \FNorm{L}^2\FNorm{R}^2.
    \end{align*}
\end{proof}


\begin{proof}[proof of Lemma~\ref{lem:distr_appro}]

    {\bf Upper bound for \textcircled{1}}. By Lemma \ref{lem:span}, $\V\x$ is a unit random vector sampled from the row space of $\A$ with probability  $1-\delta_{\V}:=(1-\delta)^{\frac{1}{4}}$ if $\epsilon_{\V}<\frac{1}{k\kappa^2}$. From Eq. \eqref{eq:orthonormal} in Lemma \ref{lem:orthonormal}, with probability $1-\delta_{\V}$
    \[
        \norm{\A\V\x-\A\hat{\V}\x} \leq \FNorm{\A}\FNorm{\V-\hat{\V}}\norm{\x} \leq (\alpha_\V/\sqrt{2} + c_1 \alpha_\V^2)\FNorm{\A}.
    \]
    Combing with $\FNorm{\A} \leq \sqrt{k}\kappa \sigma_{\min}(\A)\leq \sqrt{k}\kappa \norm{\A\V\x}$, we gain
    \begin{align}\label{ineq:1}
        \norm{\A\V\x-\A\hat{\V}\x} \leq  (\alpha_\V/\sqrt{2} + c_1 \alpha_\V^2)\sqrt{k}\kappa\norm{\A\V\x}.
    \end{align}
    Eq. \eqref{ineq:1} satisfies $\norm{\A\V\x-\A\hat{\V}\x}\le \frac{\epsilon}{8}\norm{\A\V\x}$ with $\epsilon_{\V}=O\left(\min\left\{\frac{\epsilon}{\sqrt{k}\kappa},\frac{1}{k\kappa^2}\right\}\right)$. \\

    \noindent{\bf Upper bound for \textcircled{2}}. According to Lemma \ref{lem:span},  the columns of $\hat{\U}$ span a space equal to the column space of $\A$ if $\epsilon_{\U}\le \frac{1}{k\kappa^2}$ with probability $1-\delta_{\U}:=(1-\delta)^{\frac{1}{4}}$. Let $\Pi_{\hat{\U}}$ denote the orthonormal projector to image of $\hat{\U}$ (column space of $\A$). Similarly, $\Pi_{\hat{\U}}^{\perp}$ denotes the orthonormal projector to the orthogonal space of column space of $\hat{\U}$.
    \begin{align}
      &\norm{\A\hat{\V}\x-\hat{\U}\hat{\U}^T\A\hat{\V}\x} \nonumber \\
        =&\norm{(\Pi_{\hat{\U}}+\Pi_{\hat{\U}}^\perp)\A\hat{\V}\x-\hat{\U}\hat{\U}^T\A\hat{\V}\x}\nonumber\\
        =&\norm{\Pi_{\hat{\U}}\A\hat{\V}\x-\hat{\U}\hat{\U}^T\A\hat{\V}\x} \nonumber\\
        \leq &\FNorm{\Pi_{\hat{\U}} - \hat{\U}\hat{\U}^T}\norm{\A\hat{\V}\x} \leq c_2\alpha_\U \norm{\A\hat{\V}\x},  \label{ineq:2}
    \end{align}
    based on Eq. \eqref{eq:app_proj} in Lemma \ref{lem:orthonormal}. If $\epsilon_{\U}=O\left(\min\left\{\frac{\epsilon}{k},\frac{1}{k\kappa^2}\right\}\right)$, $\norm{\A\hat{\V}\x-\hat{\U}\hat{\U}^T\A\hat{\V}\x}\le \frac{\epsilon}{8}\norm{\A\hat{\V}\x}$. \\

    \noindent{\bf Upper bound for \textcircled{3}}. When $\epsilon_{\U}$ and $\epsilon_{\V}$ are discussed above, with probability $1-\eta:=(1-\delta)^{\frac{1}{4}}$, we have
    \begin{equation}\label{ineq:a}
        \norm{\hat{\U}\hat{\U}^T\A\hat{\V}\x}\ge\left(1-\frac{\epsilon}{8}\right)\norm{\A\hat{\V}x}\ge\left(1-\frac{\epsilon}{8}\right)^2\norm{\A\V\x}.
    \end{equation}
    According to Lemma \ref{cor:inner_estimation}, we obtain
    \begin{equation}\label{ineq:b}
      \FNorm{\hat{\U}^T\A\hat{\V} - \tilde{\M}} \le \zeta \FNorm{\hat{\U}}\FNorm{\hat{\V}}\FNorm{\A}\leq \zeta k^{1.5}\kappa\sqrt{(1+\frac{\alpha_\U}{k})(1+\frac{\alpha_\V}{k})}\norm{\A\V\x}.
    \end{equation}
    Combining Eq.~\eqref{ineq:a} and Eq.~\eqref{ineq:b}, the following holds
    \begin{align}\label{ineq:3}
        &\norm{\hat{\U}\hat{\U}^T\A\hat{\V}\x - \hat{\U}\tilde{\M}\x} \nonumber\\
        \le &\FNorm{\hat{\U}}\FNorm{\hat{\U}^T\A\hat{\V} - \tilde{\M}} \nonumber\\
        \leq &\zeta k^2\kappa\sqrt{(1+\frac{\alpha_\U}{k})^2(1+\frac{\alpha_\V}{k})}\norm{\A\V\x}\nonumber\\
        \leq &\zeta k^2\kappa\sqrt{(1+\frac{\alpha_\U}{k})^2(1+\frac{\alpha_\V}{k})}/ (1-\frac{\epsilon}{8})^2\norm{\hat{\U}\hat{\U}^T\A\hat{\V}\x}.
    \end{align}
    If $\zeta=O\left(\frac{\epsilon}{k^2\kappa}\right)$, then $\norm{\hat{\U}\hat{\U}^T\A\hat{\V}\x - \hat{\U}\tilde{\M}\x}<\frac{\epsilon}{8}\norm{\hat{\U}\hat{\U}^T\A\hat{\V}\x}$ holds.\\

    \noindent{\bf Upper bound for \textcircled{4}}. Since   $\epsilon_{\U}=O\left(\min\left\{\frac{\epsilon}{k},\frac{1}{k\kappa^2}\right\}\right)$ as discussed before, directly taking usage of Lemma~\ref{lem:tang}, with probability $1-\gamma:=(1-\delta)^{\frac{1}{4s}}$ we have
    \begin{align} \label{ineq:4}
        \TV{\D{\hat{\U}\tilde{\y}}, \mathcal{O}} \le \frac{\alpha_\U}{1-\alpha_\U} <\frac{\epsilon}{4}.
    \end{align}
    Hence, Algorithm~\ref{alg:fas} generates a distribution $\mathcal{O}_i$ which satisfies $\TV{O_i, \D{\A\V\x_i}}\leq \epsilon$ for $s$ random unit vectors generated simultaneously with probability $1-\delta$.
\end{proof}

The following theorem tells us how to find the largest component of $\A\V\x_i$ from distribution $\mathcal{O}_i$.
\begin{theorem}[Restatement of  Theorem 1 in \cite{du2018quantum}]\label{thm:sample_largest}
    Let $\D{}$ be a distribution over $[m]$ and $\D{}'$ is another distribution simulating $\D{}$ with total variant error $\epsilon$. Let $\x_1, \ldots, \x_N$ be examples independently sampled from $\D{}'$ and $N_i$ be the number of examples taking value of $i$. Let $\D{\max} = \max\{\D{1},\ldots,\D{m}\}$ and $\D{secmax}=\max\{\D{1}, \ldots, \D{m}\} \backslash \D{\max}$.
    If $\D{\max}-\D{secmax}>2\sqrt{2\log(4N/\delta)/N} + \epsilon$, then, for any $\delta>0$, with a probability
    at least $1-\delta$, we have
    \[
        \arg\max_i\{N_i | 1\le i \le N \} = \arg\max_i\{p_i|1\le i \le m\}.
    \]
\end{theorem}

As mentioned in \cite{du2018quantum}, the assumption about the gap between $\D{max}$ and $\D{secmax}$ is easy to satisfy in practice. By choosing $N=\log^2m$ and $\epsilon < 2\sqrt{2\log(4\log^2m/\delta)/\log^2m}$, we have $\D{max}-\D{secmax}>4\sqrt{2\log(4\log^2m/\delta)/\log^2m}$, which will converge to zero as $m$ goes to infinity.

To estimate the number of random projections we need, we denote $p^*_i$ the probability that after random projection $\beta$, a data point $A_{(i)}$ is identified as an anchor in subspace, i.e.,
$$p^*_i = \Pr_\beta(i = \mathrm{argmax}_i\{(\A\beta)_i\}).$$
In \cite{zhou2014divide}, if $p_i^* > k/\alpha$ for a constant $\alpha$, with $s=\frac{3}{\alpha} k \log k$ random projections, all anchors can be found with probability at least $1-k\exp(-\alpha s/3k)$.



\subsubsection{Complexity and Success Probability}
Note that Algorithm~\ref{alg:fas} involves operations that query and sample from matrix $A$, $\hat{U}$ and $\hat{V}$, but those operations can be implemented in $O(\log(mn)poly(k,\kappa, 1/\epsilon))$ time. Thus, in the following analysis, we just ignore the time complexity of those operations but multiple it to the final time complexity.

The running time and failure probability mainly concentrates on lines \ref{ln:fkvv}, \ref{ln:fkvu}, \ref{ln:inner-est} and \ref{ln:rej} in Algorithm \ref{alg:fas}.
The running time of lines \ref{ln:fkvv} and \ref{ln:fkvu} are $O\left(poly(k,1/\epsilon_{\V},\log 1/\delta_{\V})\right)$ and
$O\left(poly(k,\frac{1}{\epsilon_{\U}},\log \frac{1}{\delta_{\U}})\right)$, respectively, according to Theorem~\ref{thm:FKV}.
And line \ref{ln:inner-est} takes $O\left(k^2\frac{1}{\zeta^2}\log \frac{1}{\eta}\right)$ to estimate matrix $\tilde{M}$ according to Lemma~\ref{cor:inner_estimation}. And line \ref{ln:rej} with $s$ iterations totally spends $O\left(sk^2\log \frac{1}{\gamma}poly\log m\right)$. In the perspective of failure probability, lines \ref{ln:fkvv}, \ref{ln:fkvu} and \ref{ln:inner-est} take the same failure probabilities $(1-\eta)^{\frac{1}{4}}$. And line \ref{ln:rej} takes
$(1-\eta)^{\frac{1}{4s}}$ for each iteration.

Above all, the time complexity of FAS is $O\left(poly\left(k,\kappa, \log \frac{1}{\delta},\log mn\right)\right)$. The success probability is $1-\delta$.

\section{Conclusion}
\label{sec:con}
This paper presents a classical randomized algorithm FAS which dramatically reduces the running time to find anchors of low-rank matrix.
Especially, we achieve exponential speedup when the rank is logarithmic of the input scale.
Although our algorithm running in polynomial of logarithm of matrix dimension, it still has a bad dependence on rank $k$. In the future,
we plan to improve its dependence on rank as well as analyze its noise tolerance.


\bibliographystyle{alpha}
\bibliography{reference}
\end{document}